\documentclass[11pt,letterpaper,english]{article}
\usepackage{fullpage}
\usepackage{amsmath,amssymb,amsthm}
\usepackage{algorithmic}
\usepackage[ruled,vlined,linesnumbered]{algorithm2e}
\usepackage{xspace, units}
\usepackage{textcomp}
\usepackage{cite}
\usepackage{xcolor}
\usepackage{tikz}

\newtheorem{theorem}{Theorem}
\newtheorem{proposition}[theorem]{Proposition}
\newtheorem{lemma}[theorem]{Lemma}

\newtheorem{claim}[theorem]{Claim}

\newcommand{\classP}{{\sf P}}
\newcommand{\classNP}{{\sf NP}}

\newcommand{\NN}{\ensuremath{\mathbb{N}}}

\newcommand{\RR}{\ensuremath{\mathbb{R}}}

\newcommand{\R}{{\mathcal R}}

\newcommand{\e}{{\mathrm e}}
\DeclareMathOperator*{\poly}{poly} 

\newcommand{\calD}{\ensuremath{\mathcal{D}}}
\newcommand{\calS}{\ensuremath{\mathcal{S}}}

\renewcommand{\Pr}[1]{\mbox{\rm\bf Pr}\left[#1\right]}
\newcommand{\Ex}[1]{\mbox{\rm\bf E}\left[#1\right]}
\newcommand{\Exl}[2]{\mbox{\rm\bf E}_{#1}\left[#2\right]}
\newcommand{\D}{\displaystyle}
\newcommand{\junk}[1]{}

\title{Secondary Spectrum Auctions for Symmetric and\\Submodular Bidders}

\author{%
  Martin Hoefer%
  \thanks{Department of Computer Science, RWTH Aachen University, {\tt
      mhoefer@cs.rwth-aachen.de}. Supported by DFG grant Ho 3831/3-1.}
  \and Thomas Kesselheim%
  \thanks{Department of Computer Science, RWTH Aachen University, {\tt
      kesselheim@cs.rwth-aachen.de}. Supported by DFG through UMIC
    Research Centre at RWTH Aachen University.}  
}
\date{}
\begin{document}

\maketitle
\begin{abstract}
  We study truthful auctions for secondary spectrum usage in wireless
  networks. In this scenario, $n$ communication requests need to be
  allocated to $k$ available channels that are subject to interference
  and noise. We present the first truthful mechanisms for secondary
  spectrum auctions with symmetric or submodular valuations. Our
  approach to model interference uses an edge-weighted conflict graph,
  and our algorithms provide asymptotically almost optimal
  approximation bounds for conflict graphs with a small inductive
  independence number $\rho \ll n$. This approach covers a large
  variety of interference models such as, e.g., the protocol model or
  the recently popular physical model of interference. For unweighted
  conflict graphs and symmetric valuations we use LP-rounding to
  obtain $O(\rho)$-approximate mechanisms; for weighted conflict
  graphs we get a factor of $O(\rho \cdot (\log n + \log k))$. For
  submodular users we combine the convex rounding framework
  of~\cite{DughmiSTOC11} with randomized meta-rounding to obtain
  $O(\rho)$-approximate mechanisms for matroid-rank-sum valuations;
  for weighted conflict graphs we can fully drop the dependence on $k$
  to get $O(\rho \cdot \log n)$. We conclude with promising initial
  results for deterministically truthful mechanisms that allow
  approximation factors based on $\rho$.
\end{abstract}

\thispagestyle{empty}
\setcounter{page}0
\clearpage

\section{Introduction}

The development of wireless networks crucially relies on successful
management of the frequency spectrum to provide reliable network
access. Nowadays, spectrum allocation is static -- service providers
(so-called \emph{primary users}) can obtain nation-wide licenses for
channels in governmental spectrum auctions. This practice is
inefficient and problematic: While primary users often use their
spectrum bands only in selected local areas, new and innovative
applications suffer in their development, because global licenses are
difficult to obtain or generally unavailable. A major research effort
is currently underway in computer science and engineering to overcome
this artificial scarcity and let primary users open their bands in
local areas for so-called secondary usage. \emph{Auctions} are
attractive to coordinate secondary spectrum usage, as they allow
implementing social or monetary goals in a market with self-interested
participants having private information. Interest in secondary
spectrum auctions has increased significantly in recent years
(see~\cite{Hoefer11,Gopinathan11,GopinathanREV11,Zhou08,Zhou09},
and~\cite{Berry10} for a general discussion), but the algorithmic and
strategic problems are still poorly understood.

In secondary spectrum markets, a natural regulatory goal is to
maximize social welfare, i.e., the total valuation or benefit of the
channel allocation to the secondary users. As constraint for the
allocation, the assigned channels must allow successful transmission
in the presence of interference and noise. Positioning and
interference situation is often known or can sometimes even be
observed publicly, but valuations are private information of the users
and have to be collected by the algorithm. In this process, secondary
users have an obvious incentive to manipulate the algorithm by
misreporting their valuation. In this paper, we therefore strive to
design \emph{truthful mechanisms} that allocate channels and use
payments to motivate users to reveal their values truthfully.

This scenario represents a novel and non-trivial extension of combinatorial
auctions. In combinatorial auctions we have to allocate $k$
indivisible items (channels) to $n$ bidders (users). Each bidder $v$
has a valuation $b_v(S)$ for any subset $S$ of items. The goal is to
maximize social welfare, i.e., the sum of (reported) valuations for
the assigned item sets. Secondary spectrum auctions extend this model
by allowing to give a single item/channel to \emph{multiple users} if
the set of users is feasible in terms of interference. Interference
can be modeled in various ways, and we follow the approach
of~\cite{Hoefer11} where users are vertices in a publicly known
edge-weighted conflict graph. A set of users is feasible for a channel
if they form an independent set in the graph, for a suitably defined
notion of independent set. This approach covers virtually all existing
interference models in the literature~\cite{Hoefer11,WanWASA09}. For
instance, if users are communication requests in the physical model of
interference, we can use edge weights corresponding to the affectance
between requests, and feasibility due to bounded
signal-to-interference-plus-noise ratio (SINR) is then equivalent to
having an independent set (as defined below, see
also~\cite{Hoefer11}).

Interestingly, conflict graphs resulting from popular interference
models (e.g., protocol model~\cite{Wan09} or physical
model~\cite{Kesselheim10,Hoefer11,Kesselheim11}) have a small
\emph{inductive independence number} $\rho$. The wide applicability of
this non-standard graph parameter for algorithm design is only
recently starting to be explored~\cite{Akcoglu00,Ye09,Chen10}. For our
secondary spectrum auctions it allows to bypass well-known lower
bounds of $\Omega(n^{1-\epsilon})$ for approximating independent set
and derive significantly improved guarantees based on
$\rho$~\cite{Hoefer11}. However, even in ordinary combinatorial
auctions with $\rho = 1$ any efficient algorithm can only achieve a
factor of essentially $\min\{n,\sqrt{k}\}$ unless we make additional
assumptions on the user valuations~\cite{Mirrokni08,Lehmann02}.

\subsection{Contribution}

In this paper, we design randomized auctions for spectrum markets,
where secondary users strive to acquire one or more of a set of
available channels. Marking additional assumptions on the user
valuations allows us to bypass the $\min\{n,\sqrt{k}\}$ lower bound
and to significantly improve previous results. We examine the two
prominent classes of symmetric and submodular valuations. Both classes occupy a central position in the
literature on combination auctions, and they have very natural and
intuitive intepretations in the context of secondary spectrum
auctions.

Symmetric valuations are the analog to multi-unit auctions, where each
valuation only depends on the \emph{number} of channels rather than
the exact subset. This is a natural assumption in a secondary spectrum
auction of equally sized channels which all offer very similar
conditions. Submodularity is economically interpreted as diminishing
marginal returns. A common representative are coverage valuations,
where users pick elements each covering a certain range, and the value
is the total covered area. This is a natural assumption, e.g, when
secondary users are transmitters that strive to be received by as many
mobile stations as possible, where each of the latter operates on a
fixed subset of channels.

For symmetric valuations (see Section~\ref{sec:symmetric}) we use the
intuition of multi-unit auctions and round a suitably defined linear
program yielding only an assignment of numbers of channels. Using
these numbers an independent set for each channel is then created by a
greedy approach. This allows to avoid dependence on $k$ and obtain an
approximation factor of $O(\rho)$ for unweighted conflict graphs. Note
that this is asymptotically almost optimal under standard complexity
assumptions. Theorem~5 in \cite{Hoefer11} shows that there is no $\rho
/ 2^{O(\sqrt{\log \rho})}$-approximation unless $\classP =
\classNP$. Truthfulness is achieved via combination of our approach
with the celebrated randomized meta-rounding framework by Lavi and
Swamy~\cite{Lavi05}. For edge-weighted conflict graphs, the
construction step of independent sets is significantly more involved. The asymmetry of conflicts inherent in edge-weighted graphs
require the use of additional concurrent contention resolution methods
to partition the rounded set of requests into feasible independent
sets. This approach allows to obtain a factor of $O(\rho\cdot(\log n +
\log k))$. Our resulting mechanisms are randomized, run in polynomial
time, and yield truthfulness in expectation.

For submodular valuations (see Section~\ref{sec:mrs}) we focus on
matroid-rank-sum valuations, which encompass the most frequently
studied submodular valuations. We design randomized mechanisms that
fall into the class of maximum-in-distributional range (MIDR)
mechanisms. In particular, our approach is along the lines of the
convex rounding technique recently pioneered
in~\cite{Dughmi11,DughmiSTOC11} and achieves an approximation factor
of $O(\rho)$ for unweighted conflict graphs. Again, this is
asymptotically almost optimal under standard complexity
assumptions. In contrast to the case of symmetric valuations, we can
fully omit the dependence on $k$ and show factors of $O(\rho \cdot
\log n)$ even for weighted conflict graphs. Our rounding scheme is
similar to the Poisson rounding scheme from~\cite{DughmiSTOC11}. The
main difference and complication is again the need to round each
channel to an independent set of users. To achieve this, we round
independently for each channel and build the required support of
independent sets using a randomized meta-rounding approach. Probably
the most technical contribution is showing that this rounding scheme
preserves the favorable conditioning properties that allow to apply convex optimization techniques to compute the
underlying distribution with sufficient precision in expected
polynomial time, even for weighted conflict graphs. Our resulting
mechanisms are again randomized and provide truthfulness in
expectation.

Finally, we also briefly discuss designing deterministic truthful
mechanisms (see Section~\ref{sec:discuss}). We present a promising
initial result, a monotone greedy $O(\rho \cdot \log n)$-algorithm for
a single channel in unweighted conflict graphs. However, this area
remains mostly as an interesting and important avenue for future work.

\subsection{Related Work}
\label{sec:contribution}

Our paper is connected to recent approaches for designing truthful
mechanisms in secondary spectrum markets without~\cite{Zhou08,Zhou09}
and with non-trivial worst-case approximation guarantees, e.g., for
social welfare and fairness~\cite{Gopinathan11} or
revenue~\cite{GopinathanREV11}. However, all these works are
restricted to a single channel and unweighted conflict graphs. To this
date the only general (analytical) treatment of approximation
algorithms and truthful mechanisms for multi-channel secondary
spectrum auctions is~\cite{Hoefer11} where truthful-in-expectation
mechanisms for general user valuations are designed using the
inductive independence number in edge-weighted conflict graphs. For
unweighted conflict graphs the approximation guarantee is $O(\rho
\cdot \sqrt{k})$, for edge-weighted conflict graphs $O(\rho \cdot
\sqrt{k} \cdot \log n)$. The former result is asymptotically almost
optimal in $\rho$ if $k=1$~\cite{Trevisan01} and in $k$ if $\rho =
1$. The latter lower bound is a well-known result in combinatorial
auctions~\cite{Mirrokni08,Lehmann02}.

In ordinary combinatorial auctions, these strong lower bounds
initiated the study of relevant subclasses of valuations, for an
overview see, e.g.,~\cite{BlumrosenChapter07}. Symmetric valuations
essentially pose a knapsack problem of assigning numbers of items to
bidders, and a deterministic truthful greedy
2-approximation~\cite{Mualem02} was the first benchmark
solution. Since then there has been significant progress including,
e.g., approximation schemes for single-minded bidders~\cite{Briest05},
$k$-minded bidders~\cite{Dobzinski10}, or monotone
valuations~\cite{Dobzinski09,Voecking12}. In contrast to these works,
we must additionally decompose assigned numbers of channels into an
independent set for each single channel. Here we rely on rounding
linear programs to ensure that such a decomposition exists and can be
found in polynomial time.

For submodular valuations, social welfare maximization without
truthfulness is essentially solved. Optimal $(1-1/\e)$-approximation
algorithms exist even for value oracle access~\cite{Vondrak08}, where
each valuation $b_v$ is an oracle that we can query to obtain $b_v(S)$
for a single set $S$ in each operation. This factor cannot be improved
assuming either polynomial communication~\cite{Mirrokni08} in the
value oracle model or polynomial-time complexity in
general~\cite{Khot08}. For the strategic setting and general
submodular valuations, the best factors are $O\left(\frac{\log m}{\log
    \log m}\right)$ for truthfulness in
expectation~\cite{DobzinskiCoRR10}, and $O(\log m \log \log m)$ for
universal truthfulness~\cite{Dobzinski07}. Dughmi et
al~\cite{DughmiSTOC11} recently proposed a convex rounding technique
to build truthful-in-expectation mechanisms. Their approach yields an
optimal $(1-1/\e)$-approximation for the class of matroid-rank-sum
valuations. It follows the idea of maximal-in-distributional range
(MIDR) mechanisms by defining a range of distributions independent of
the valuations and a rounding procedure. Both are designed in a way
that finding the optimal distribution over the range for the reported
valuations becomes a convex program with favorable conditioning
properties. Hence, the optimal distribution can be found using
suitable convex optimization methods in expected polynomial
time. Truthfulness follows using the Vickrey-Clarke-Groves (VCG)
payment scheme. Very recently, Dughmi and Vondrak showed that a
similar result cannot be obtained for general submodular valutions in
the oracle model~\cite{DughmiFOCS11}.

Designing (non-truthful) algorithms for independent set problems in
conflict graphs has received significant attention recently,
especially for graphs based on the physical model of interference with
SINR constraints. If each request has a value of 1 for being in the
independent set, asymptotically optimal performance bounds for
specific transmission power assignments were obtained when requests
are located in various classes of metric spaces~\cite{Goussevskaia09,
  Fanghaenel10, Halldorsson11}. For the problem where powers can be
arbitrarily chosen, there is a constant-factor approximation
algorithm~\cite{Kesselheim11}.

The inductive independence number is a non-standard graph parameter
that is only recently starting to receive increased attention. Up to
our knowledge the parameter has first been used in~\cite{Akcoglu00},
and since then has been rediscovered independently a number of times
(see, e.g., ~\cite{Wan09}). Ye and Borodin~\cite{Ye09} recently
conducted the first study addressing general issues that arise when
using the measure for solving algorithmic problems in unweighted
graphs. The eminent usefulness of the parameter for analyzing
interference models and spectrum markets was highlighted
in~\cite{Hoefer11}.

\section{Preliminaries}

\subsection{Channel Allocation in Spectrum Markets}

In secondary spectrum markets there is a set $[k]$ of $k$ available
\emph{channels} and a set $V$ of $n$ \emph{users} or
\emph{bidders}. Each user $v \in V$ has a \emph{valuation} or
\emph{benefit} $b_v : 2^{[k]} \to \RR^+$. A valuation function $b_v$
is called \emph{symmetric} if $b_v(T) = b_v(|T|)$ for all $T \subseteq
[k]$. It is \emph{submodular} if $b_v(T \cup T') + b_v(T \cap T') \le
b_v(T) + b_v(T')$ for all $T \subseteq T'$. For submodular valuations
we also assume they are \emph{monotone} with $b_v(T) \le b_v(T')$ for
$T \subseteq T'$.  A valuation $b_v$ is a \emph{matroid rank sum
  (MRS)} function if there exists a family of matroid rank functions
$u_1, \ldots, u_{\kappa} : 2^{[k]} \to \NN$, and associated
non-negative weights $w_1, \ldots, w_{\kappa} \in \RR^+$, such that
$b_v(T) = \sum_{\ell = 1}^{\kappa} w_\ell u_\ell(T)$ for all $T
\subseteq [k]$.

To model interference we represent users as vertices in a complete
edge-weighted and directed \emph{conflict graph} $G = (V,E,w)$. The
weight $w(u,v)$ of edge $(u,v)$ represents the interference that user
$u$ creates for user $v$ if both are assigned to the same
channel. Interference between users is similar on each channel. A set
of users $U \subseteq V$ is \emph{feasible} or an \emph{independent
  set} if $\sum_{u \in U} w(u,v) < 1$ for all $v \in U$. In
\emph{unweighted} conflict graphs all weights $w(u,v) \in \{0,1\}$ and
our definition of independent set is the same as in the classical
sense. For many standard interference models, we can define weighted
conflict graphs such that independent sets are exactly the sets for
which we can have successful simultaneous transmission in the
interference model. For instance, the protocol model results in
unweighted conflict graphs, or the physical model of interference
yields weighted conflict graphs where independent sets are feasible
with respect to the SINR; for details see~\cite{Hoefer11}.

The algorithmic challenge in secondary spectrum markets is the
\emph{channel allocation problem}. In an optimal solution $S$, each
user $v$ receives a subset of channels $S_v \subseteq [k]$ such that
each channel is given to an independent set in the conflict graph and
the \emph{social welfare} $b(S) = \sum_{v \in V} b_v(S_v)$ is
maximized. In contrast to ordinary combinatorial auctions, an
independent set can include more than one user. Our mechanisms cope
with this issue using a structural parameter called inductive
independence number. Let us define \emph{symmetric weights} by
$\bar{w}(u,v) = w(u,v) + w(v,u)$. Then the \emph{inductive
  independence number} is the smallest number $\rho$ such that there
is an ordering $\pi$ of the vertices satisfying the following
condition: For all $v \in V$ and all independent sets $M \subseteq V$,
we let $M_v = M \cap \{u \in V \mid \pi(u) < \pi(v)\}$ and have that
$\sum_{u \in M_v} \bar{w}(u,v) \le \rho$. Hence, $\rho$ is the
smallest number such that by picking the best ordering we can bound
for any $v \in V$ the incoming weight from any independent set among
previous vertices to at most $\rho$. We assume that $\rho$ and the
ordering $\pi$ of $V$ are given. For many interference models and
their resulting conflict graphs we can find in polynomial time small
upper bounds on $\rho$ and a corresponding ordering witnessing
$\rho$. For example, in the protocol model $\rho = O(1)$~\cite{Wan09}
and in the physical model $\rho = O(\log n)$~\cite{Kesselheim10} or
$\rho = O(1)$~\cite{Kesselheim11}, depending on power control
assumptions. In both cases, $\pi$ orders users with decreasing or
increasing distance between sender and receiver.

\subsection{Mechanism Design Basics}

To avoid that user $v$ will strategically misreport his valuation, we
charge payments $p_v$ and make \emph{truthfulness} a dominant
strategy. For each user $v \in V$ we ensure that his
\emph{quasi-linear utility} satisfies $b_v(S_v) - p_v(b_v,b_{-v}) \ge
b_v(S'(v)) - p_v(b_v',b_{-v})$, where $S$ and $S'$ are our solutions
to the channel allocation problem when $v$ reports the true $b_v$ and
a some possibly other $b_v'$, respectively. This can be achieved using
classic \emph{Vickrey-Clarke-Groves (VCG) payments} if the allocation
problem is always solved optimally.

In contrast, efficient truthful mechanisms cannot compute optimal
solutions to intractable problems. For some problems, deterministic
mechanisms can achieve only trivial approximation
guarantees~\cite{Papadimitriou08}. The situation is much better if we
resort to \emph{randomized mechanisms}, which define a distribution
$D$ over the set of solutions \calS\ for the channel allocation
problem and output an allocation $S \in \calS$ according to $D$. In
this case, we aim for \emph{truthfulness in expectation}, i.e., for
every $v \in V$
\[
\Exl{S \sim D}{b_v(S_v) - p_v(b_v,b_{-v})} \ge \Exl{S \sim
  D'}{b_v(S_v) - p_v(b_v',b_{-v})}\enspace,
\]
where $D'$ is the distribution if $v$ reports $b_v'$ instead of
$b_v$. A general technique to design such mechanisms is
\emph{maximal-in-distributional range (MIDR)}. Here we fix a set (the
range) of distributions \calD\ over \calS, where \calD\ is independent
of the valuations $b_v$. The algorithm receives all reported
valuations $b_v$ and optimizes exactly over \calD\ to find $D \in
\calD$ with maximum expected social welfare. Due to exact optimization
over \calD, the mechanism can use VCG payments to guarantee
truthfulness in expectation. The obvious problem in MIDR is designing
the distributional range \calD\ (1) large enough to contain a good
approximation for every possible vector of user valuations, and (2)
small enough to allow for exact optimization over \calD\ in polynomial
time. Our mechanisms in Sections~\ref{sec:symmetric} and~\ref{sec:mrs}
will all be MIDR mechanisms. In Section~\ref{sec:discuss} we also
briefly treat designing greedy mechanisms that are truthful and
deterministic.

\section{Symmetric Valuations}
\label{sec:symmetric}

In this section we consider spectrum auctions with \emph{symmetric
  valuations} in which $b_v(T) = b_v(\lvert T \rvert)$ for all $v \in
V$. We concentrate on designing approximation algorithms that can be turned
into truthful MIDR mechanisms following the framework by Lavi and
Swamy~\cite{Lavi05}.

Our algorithms round the following LP relaxation based on $k \cdot
\lvert V \rvert$ variables $x_{v, i} \in \{0,1\}$ indicating if $v$
gets exactly $i$ channels or not. The relaxation reads
\begin{equation}
  \label{lp:MultiUnit}
  \begin{array}{lrcll}
    \mbox{Max. } & \D \sum_{v \in V} \sum_{i=1}^k b_v(i) \cdot x_{v, i}\\
    \mbox{s.t. } & \D \sum_{\substack{u \in V\\\pi(u) < \pi(v)}} \;\; \sum_{i=1}^k i \cdot \bar{w}(u, v) \cdot x_{u, i} &\leq& \rho \cdot k & \text{for all $v \in V$} \\
    & \D \sum_{i=1}^k x_{v, i} &\leq& 1 & \text{for all $v \in V$} \\
    & x_{v,i} &\ge& 0 & \text{for all $v \in V$, $i \in [k]$.}
  \end{array}
\end{equation}
Note that this relaxation does not describe the problem exactly, as an
integral solution to the relaxation might not be feasible for the
channel allocation problem. In particular, the relaxation does not
specify which user receives which channel, but this information is
critical for interference and feasibility of the requests.

We solve the LP relaxation optimally. The computed fractional solution
is then decomposed into two solutions $x^{(1)}$ and $x^{(2)}$, that
are rounded separately. Based on such a solution, for each user $v$ a
preliminary number of channels $d^{(l)}_v$ is determined at
random. The probability is proportional to the fractional variables
$x^{(l)}_{v, i}$. Having assigned these numbers of channels, we still
have to derive a feasible allocation. In this allocation, each user
$v$ either gets $d^{(l)}_v$ channels or none.

\subsection{Unweighted Conflict Graphs}

In the case of unweighted conflict graphs, we use a simple greedy
approach to distribute available channels to users, see
Algorithm~\ref{algo:MultiUnit}.
\begin{algorithm}[t]
  \label{algo:MultiUnit}
  \DontPrintSemicolon
  Decompose an optimal solution $x$ to LP~\eqref{lp:MultiUnit} into
  two solutions $x^{(1)}$ and $x^{(2)}$ as follows:\hspace{2cm} Set
  $x^{(1)}_{v,i} = x_{v,i}$ if $i \leq k/2$ and $x^{(1)}_{v,i} = 0$
  otherwise; set $x^{(2)} = x - x^{(1)}$. \;

  \For{$l \in \{1, 2\}$}{ 
    \For{$v \in V$ in increasing order of $\pi$ values}{ 
      with probability $\frac{x^{(l)}_{v, i}}{4 \rho}$ set $d^{(l)}_v
      := i$\; 
      Let $F^{(l)}_v := \{i \in [k] \mid \text{there is no $u \in
        \Gamma_\pi(v)$ with $i \in S^{(l)}_v$} \}$ \; 
      
      $S^{(l)}_v= \begin{cases} \mbox{arbitrary $M \subseteq F^{(l)}_v$ with $|M| = d^{(l)}_v$} & \mbox{if } \lvert F^{(l)}_v \rvert \geq d^{(l)}_v, \\
        \emptyset & \mbox{otherwise} \end{cases}$
      
    }
  }
  Return the better one of the solutions $S^{(1)}$ and $S^{(2)}$\;
  \caption{LP-Rounding for Symmetric Valuations and Unweighted
    Conflict Graphs}
\end{algorithm}
The expected social welfare of the output will decrease only by a
factor of $O(\rho)$ under the fractional optimum, which is
asymptotically optimal.
\begin{theorem}
  Algorithm~\ref{algo:MultiUnit} returns a feasible allocation of
  social welfare at least $\nicefrac{b^\ast}{16 \rho}$ in expectation.
\end{theorem}

\begin{proof}
  Solutions $S^{(1)}$, $S^{(2)}$ separate the problem into two
  subproblems, in which the maximum or minimum non-zero number of
  channels allocated to a single player is $k/2$, respectively. We
  analyze both of these cases separately in the key proposition.

  \begin{proposition}
    For $l \in\{1,2\}$ and the expected social welfare of $S^{(l)}$ we
    have
    \[ \Ex{b(S^{(l)})} \geq \frac{1}{8 \rho} \cdot \sum_{v \in V}
    \sum_{i=1}^k b_v(i) \cdot x^{(l)}_{v, i}\enspace.\]
  \end{proposition}
  
  \begin{proof}
    For all $v \in V$, $i \in [k]$, $l \in \{0,1\}$ let $X_{v,
      i}^{(l)}$ be a 0/1 random variable indicating if in the rounding
    stage $d^{(l)}_v$ is set to $i$. We know that $\Pr{X^{(l)}_{v, i}
      = 1} = \frac{x^{(l)}}{4 \rho}$. Let $Y_{v, i}^{(l)}$ be a 0/1
    random variable indicating if $\lvert S^{(l)}_v \rvert = i$. To
    show the proposition it remains to bound $\Pr{Y_{v, i}^{(l)} = 0
      \mid X_{v, i}^{(l)} = 1}$; that is, the probability that a user
    $v$ does not receive $i$ channels although $d^{(l)}_v$ was set to
    $i$.
    
    \begin{description}
    \item[Case $l=1$:] The event that $Y^{(1)}_{v, i} = 0$ but
      $X^{(1)}_{v, i} = 1$ can only occur if $\lvert F^{(1)}_v \rvert
      \leq i$. So in particular $\lvert F^{(1)}_v \rvert \leq k /
      2$. We can express $\lvert F^{(1)}_v \rvert$ in terms of $Y_{v,
        i}^{(l)}$ as
    \[
    k - \lvert F^{(1)}_v \rvert \leq \sum_{u \in \Gamma_\pi(v)}
    \sum_{i=1}^k i \cdot Y_{u, i}^{(1)} \leq \sum_{u \in
      \Gamma_\pi(v)} \sum_{i=1}^k i \cdot X_{u, i}^{(1)}\enspace.
    \]
    By linearity of expectation and the definition of $\rho$ this
    yields
    \[
    \Ex{k - \lvert F^{(1)}_v \rvert} \leq \sum_{u \in \Gamma_\pi(v)}
    \sum_{i=1}^k i \cdot \frac{x^{(1)}_{u, i}}{4 \rho} \leq
    \frac{k}{4}\enspace.
    \]
    So, we get by Markov inequality
    \[
    \Pr{\lvert F^{(1)}_v \rvert \leq \frac{k}{2}} = \Pr{k - \lvert
      F^{(1)}_v \rvert \geq \frac{k}{2}} \leq \frac{1}{2}\enspace.
    \]
    In total this yields
    \[
    \Pr{Y^{(1)}_{v, i} = 1} = \Pr{X_{v, i}^{(1)} = 1} \cdot \Pr{\lvert
      F^{(1)}_v \rvert \geq i} \geq \Pr{X_{v, i}^{(1)} = 1} \cdot
    \Pr{\lvert F^{(1)}_v \rvert \geq \frac{k}{2}} \geq
    \frac{x^{(1)}_{v, i}}{8 \rho}\enspace,
    \]
    which proves the proposition in Case 1.

  \item[Case $l=2$:] The event that $Y_{v, i}^{(2)} = 0$ but $X_{v,
      i}^{(2)} = 1$ can only happen if there is some $u \in
    \Gamma_\pi(v)$ with $S^{(2)}_u \neq \emptyset$, in which case $\D
    \sum_{u \in \Gamma_\pi(v)} \sum_{i = k/2 + 1}^k Y_{u, i}^{(2)}
    \geq 1$. Furthermore, we have
    \[
    \sum_{u \in \Gamma_\pi(v)} \sum_{i = k/2 + 1}^k Y_{u, i}^{(2)}
    \leq \sum_{u \in \Gamma_\pi(v)} \sum_{i = k/2 + 1}^k X_{u,
      i}^{(2)} \leq \frac{2}{k} \sum_{u \in \Gamma_\pi(v)} \sum_{i =
      k/2 + 1}^k i \cdot X_{u, i}^{(2)}
    \]
    Using linearity of expectation and the definition of $\rho$ this
    yields
    \[
    \Ex{ \sum_{u \in \Gamma_\pi(v)} \sum_{i = k/2 + 1}^k Y_{u,
        i}^{(2)} } \leq \frac{2}{k} \sum_{u \in \Gamma_\pi(v)} \sum_{i
      = k/2 + 1}^k i \cdot \frac{x^{(2)}_{u, i} }{4 \rho} \leq
    \frac{2}{k}\cdot\frac{k}{4} \le \frac{1}{2}\enspace,
    \]
    Markov inequality then implies that the probability that all of
    the $u \in \Gamma_\pi(v)$ have $S^{(2)}_u = \emptyset$ is at
    least $1/2$. This means
    \begin{eqnarray*}
      \Pr{Y^{(2)}_{v, i} = 1} &=& \Pr{X_{v, i}^{(2)} = 1} \cdot \Pr{\lvert
        F^{(1)}_v \rvert \geq i}\\
      &\geq& \Pr{X_{v, i}^{(2)} = 1} \cdot
      \Pr{\forall u \in \Gamma_\pi(v), S^{(2)} = \emptyset} \geq
      \frac{x^{(1)}_{v, i}}{8 \rho}\enspace,
    \end{eqnarray*}
    which proves the Proposition for Case 2.
  \end{description}
\end{proof}

Finally, to prove the theorem we note that by splitting the solution
into two parts and returning the better output, we lose only a factor
of 2 in the approximation guarantee. For the expected social welfare
it holds $\max_{l \in \{1,2\}} \sum_{v \in V} \sum_{i=1}^k b_v(i)
\cdot x^{(l)}_{v, i} \ge b^\ast/2.$
\end{proof}

\subsection{Edge-Weighted Conflict Graphs}
Allocating the channels is much more involved in the case of
edge-weighted conflict graphs due to the asymmetry of interference
constraints. In the unweighted case the simple greedy allocation only
has to make sure there are no edges to vertices on the same
channel. This is unsuitable now since adding a user might violate
constraints at previously added users -- even though constraints are
satisfied for the currently added user.

Having obtained the $d^{(l)}_v$ values in the described way, we first
consider only the incoming weight from users of smaller index like in
the unweighted case. If the incoming weight from previous users is too
high, i.e., $\sum_{u \in V, \pi(u) < \pi(v)} d^{(l)}_u \cdot
\bar{w}(u, v) \geq k / 32$, we remove all channels from the user and
set $d^{(l)}_v := 0$. However, unlike in the unweighted case, this
does not yet guarantee the existence of an allocation. The crucial
difference occurs in the last step, where the allocation is
derived. This step is performed differently for the two solutions of
the decomposition. For the case in which each user was assigned at
most $k/8$ channels, the allocation is made in a randomized fashion in
Algorithm~\ref{algo:WeightedMultiUnitAllocateSmall}. For the other
case, the allocation is made deterministically in
Algorithm~\ref{algo:WeightedMultiUnitAllocateLarge}. Unlike in the
unweighted case, in both cases the resulting allocation will not
include all users at a time but only allocate channels to a subset of
the originally chosen users.

\begin{algorithm}
  \DontPrintSemicolon
  Decompose an optimal solution $x$ to LP~\eqref{lp:MultiUnit} into
  two solutions $x^{(1)}$ and $x^{(2)}$ as follows:\hspace{2cm} Set
  $x^{(1)}_{v,i} = x_{v,i}$ if $i \leq k/8$ and $x^{(1)}_{v,i} = 0$
  otherwise; set $x^{(2)} = x - x^{(1)}$. \;

  \For{$l \in \{1, 2\}$}{ 
    \For{$v \in V$ in increasing order of $\pi$ values}{ 
      With probability $\frac{x^{(l)}_{v, i}}{64 \rho}$ set $d^{(l)}_v
      := i$\; 
      Set $d^{(l)}_v := 0$ if $\sum_{u \in V, \pi(u) < \pi(v)} d^{(l)}_u \cdot \bar{w}(u, v) \geq k / 32$\;
    }
    Run Algorithm \textsc{Allocate}($l$) on $d^{(l)}$, let $S^{(l)}$ be the result\; 
  }
  Return the better one of the solutions $S^{(1)}$ and $S^{(2)}$\;
  \caption{LP-Rounding for Symmetric Valuations and Weighted
    Conflict Graphs \label{algo:WeightedMultiUnitLPRounding}}
\end{algorithm}
\begin{theorem}
  Algorithm~\ref{algo:WeightedMultiUnitLPRounding} returns a feasible
  allocation of social welfare at least $\Omega(\nicefrac{b^\ast}{\rho
    \cdot(\log n + \log k)})$ in expectation.
\end{theorem}

In order to show the bound, we will show that both LP solutions are
rounded to feasible allocations that are in expectation at most a
$O(\rho \cdot (\log n + \log k))$ factor worse than the respective LP
solution.
  
As a first step, we analyze the input given in terms of the number of
channels for each user. In particular, we show that an allocation
satisfying all of these demands simultaneously would in expectation be
at most a $1 / 128 \rho$ factor worse than the fractional solution.
\begin{proposition}
  For $l \in\{1,2\}$ and the expected social welfare of $d^{(l)}$ we
  have
  \[
  \Ex{ \sum_{v \in V} b_v(d^{(l)}_v) } \geq \frac{1}{128 \rho} \cdot
  \sum_{v \in V} \sum_{i=1}^k b_v(i) \cdot x^{(l)}_{v, i}\enspace.
  \]
\end{proposition}
\begin{proof}
  For all $v \in V$, $i \in [k]$, $l \in \{0,1\}$ let $X_{v,i}^{(l)}$
  be a 0/1 random variable indicating if in the rounding stage
  $d^{(l)}_v$ is set to $i$. We know that $\Pr{X^{(l)}_{v, i} = 1} =
  \frac{x^{(l)}}{4 \rho}$. Let $Y_{v, i}^{(l)}$ be the respective 0/1
  random variable at the time when the allocation algorithm is
  started.

  We have to bound $\Pr{Y_{v, i}^{(l)} = 0 \mid X_{v, i}^{(l)} =
    1}$. This is the probability that the weight bound in line 5 is
  exceeded. By Markov inequality, we get
  \[
  \Pr{Y_{v, i}^{(l)} = 0 \mid X_{v, i}^{(l)} = 1} \leq \frac{32}{k}
  \cdot \Ex{ \sum_{u \in V, \pi(u) < \pi(v)} d^{(l)}_u \bar{w}(u, v)
    X^{(l)}_{v, i} }\enspace.
  \]
  Applying linearity of expectation and the fact we have an LP
  solution this is
  \[
  \frac{32}{k} \cdot \sum_{\substack{u \in V\\ \pi(u) < \pi(v)}} d^{(l)}_u \cdot
  \bar{w}(u, v) \cdot \frac{x^{(l)}_{v, i}}{64 \rho} \quad \leq \quad \frac{1}{2}
  \enspace.
  \]

  In total, we obtain
  \[
  \Ex{ \sum_{v \in V} b_v(d^{(l)}_v) } = \sum_{v \in V} \sum_{i=1}^k
  b_v(i) \cdot \Pr{Y_{v, i}^{(l)} = 1} \geq \frac{1}{128 \rho} \sum_{v \in V} \sum_{i=1}^k
  b_v(i) \cdot x^{(l)}_{v,i}
  \enspace.
  \]
\end{proof}

In the two following subsections, we consider the two allocation
algorithms and show that in either case a feasible allocation of
social welfare at least $\Omega( \sum_{v \in V} b_v(d^{(l)}_v) /
(\log n + \log k) )$ is computed.

\subsubsection{{\sc Allocate}(1): Allocation algorithm for ``small'' sets}

From a preliminary selection of numbers of channels
Algorithm~\ref{algo:WeightedMultiUnitAllocateSmall} generates a
feasible allocation in which $d_v \leq k / 8$ for each $v \in V$ and
$\sum_{u \in V, \pi(u) < \pi(v)} d_u \bar{w}(u, v) < k / 32$. The idea
is that a number of allocations are computed having the property that
each user is considered in exactly one of these allocations. Each
allocation is computed by first selecting a subset of all users and
then performing $k$ randomized contention resolution steps. We iterate
over the $k$ channels, and for each channel we let each user $v$
independently perform a random experiment. With probability $8 d_v /
k$ it receives this channel tentatively. If the user received
$d_v$ channel it keeps the respective channels in this
allocation is dropped from consideration. All other users are
allocated in later rounds. The main argument to show that this yields
feasibility and provides the desired bound on the approximation factor
relies on a suitable tracking of the degrees during the contention
resolution process.

\begin{algorithm}
Set $V_0 := V$ and $t := 0$\;
\While{$V_t \neq \emptyset$}{
\For{$u \in V_t$ in decreasing order of $\pi$ values}{
  \If{$\sum_{v \in H_t} d_v \cdot \bar{w}(u, v) < k / 32$}{
    Add $u$ to $H_t$ and for each $j \in [k]$ set $X_{v, j}$ independently 
    to $1$ with probability $8 d_v / k$\;
  }
}
\For{$v \in H_t$}{
  For each $j \in [k]$ set $Y_{v,j} = 1$ if $\sum_{u \neq v} \bar{w}(u,v) \cdot X_{u, j} < 1$\;
  \If{$\sum_{j \in [k]} Y_{v,j} \ge d_v$}{
    set $S_v^t$ to an arbitrary subset of $d_v$ channels $j$ with $Y_{v,j} =1$\;
  }
}
Let $V_{t+1}$ be the set of users who have not been allocated anything and set $t := t + 1$\;
}
Return the best one of the allocations $S^1, S^2, S^3, \ldots$\;
\SetAlgoRefName{{\sc Allocate}(1)}
\caption{Channel allocation for users that require at most $k / 8$
  channels.\label{algo:WeightedMultiUnitAllocateSmall} }
\end{algorithm}

\begin{lemma}
  The allocation has social welfare at least $\Omega(\sum_{v \in V}
  b_v(d_v) / (\log n + \log k) )$ with high probability, i.e., with
  probability at least $1 - (nk)^{-c}$ for any constant $c > 1$.
\end{lemma}

\begin{proof}
  In order to show this bound, it suffices to prove that
  \[
  \Ex{ \sum_{v \in V_{t + 1}} d_v \mid V_t } \leq \frac{3}{4} \sum_{v
    \in V_t} d_v \enspace.
  \]
  Using Markov inequality this implies that for each constant $c>1$
  the probability that the set $V_t$ with $t = (c+1)\log(nk) /
  \log(4/3)$ is not empty is at most
  \[
  \Pr{\sum_{v \in V_t} d_v \geq 1} \leq \Ex{\sum_{v \in V_t} d_v} \leq
  \left( \frac{3}{4} \right)^t nk = (nk)^{-c}\enspace.
  \]
  Thus with high probability at most $O(\log(\sum_{v \in V} d_v)) =
  O(\log n + \log k)$ allocations are computed.

  We prove the bound in two steps. First, we show that the sum of
  demands in the set $H_t$ is at least half of the total demands in
  $V_t$. Afterwards, we observe that for a user in $H_t$, the
  probability to be included is at least $\frac{1}{2}$.

  \begin{claim} 
    \label{clm:log}
    \[
    \sum_{v \in H_t} d_v > \frac{1}{2} \sum_{v \in V_t} d_v \enspace.
    \]
  \end{claim}
  
  \begin{proof}
    Each user $u \in V_t \setminus H_t$ was excluded from $H_t$
    because we have
    \[
    \sum_{\substack{v \in H_t \\ \pi(u) < \pi(v)}} d_v \cdot
    \bar{w}(u,v) \geq \frac{k}{32} \enspace.
    \]
    Taking the sum, weighted by the respective $d_u$ value, we get
    \[
    \sum_{u \in V_t \setminus H_t} d_u \cdot \sum_{\substack{v \in H_t \\
        \pi(u) < \pi(v)}} d_v \cdot \bar{w}(u,v) \geq 
    \sum_{u \in V_t \setminus H_t} d_u \cdot \frac{k}{32} \enspace.
    \]
    On the other hand, we have
    \[
    \sum_{u \in V_t \setminus H_t} d_u \cdot \sum_{\substack{v \in H_t \\
        \pi(u) < \pi(v)}} d_v \cdot \bar{w}(u,v) = \sum_{v \in H_t}
    d_v \cdot \sum_{\substack{u \in V_t \setminus H_t \\ \pi(u) <
        \pi(v)}} d_u \cdot \bar{w}(u,v) < \sum_{v \in H_t} d_v \cdot
    \frac{k}{32} \enspace.
    \]
    Assembling the two bounds yields the claim.
  \end{proof}

  \begin{claim}
    The probability for each user $v \in H_t$ to be included in the
    allocation is at least $\frac{1}{2}$.
  \end{claim}
  
  \begin{proof}
    A user $v \in V$ is not included in the allocation if there is a set
    $M \subseteq [k]$ with $\lvert M \rvert \geq k - d_v$ such that
    $Y_{v, j} = 0$ for all $j \in M$.
    
    Let us first consider a single channel $j$. In order to have
    $Y_{v, j} = 1$, two independent events have to occur: First, we
    have to have $X_{v, j} = 1$ and second $\sum_{u \neq v} \bar{w}(u,
    v) \cdot X_{u, j} < 1$. The probability for the first one is defined in
    the algorithm, the second one can be bounded by the Markov
    inequality to get
    \[
    \Pr{Y_{v, j} = 1} \geq \Pr{X_{v, j} = 1} \cdot \left( 1 - \Ex{
        \sum_{u \neq v} \bar{w}(u, v) \cdot X_{u, j} } \right) = \frac{8
      d_v}{k} \cdot \left( 1 - \sum_{u \neq v} \bar{w}(u, v) \cdot \frac{8
        d_u}{k} \right) \geq \frac{4 d_v}{k} \enspace.
    \]

    Now consider a block $B$ of $\lfloor \frac{k}{2 d_v} \rfloor \geq
    \frac{3 k}{8 d_v}$ consecutive channels. Since the random
    experiments are independent, for such a block $B$ the probability
    of $\sum_{j \in B} Y_{v, j} = 0$ is at most
    \[
    \left( 1 - \frac{4 d_v}{k} \right)^{\frac{3 k}{8 d_v}} \leq
    \exp\left( - \frac{3}{2} \right) \leq \frac{1}{4} \enspace.
    \]

    Since there are $k$ channels in total, we have at least $2 d_v$
    blocks in total. For each of these blocks, the probability of $v$
    getting no channel in this block is at most $\frac{1}{4}$. This
    is, the expected number of blocks $B$ in which $\sum_{j \in B}
    Y_{v, j} = 0$ is at most $\frac{d_v}{2}$. Using the Markov
    inequality, the probability that there are more than $d_v$ blocks
    without a channel for $v$ is less than $\frac{1}{2}$. Thus, with
    probability at least $\frac{1}{2}$, $v$ gets at least 1 channel in
    at least $d_v$ blocks. This yields the claim.
  \end{proof}

  Combining these two insights, we get the desired bound which proves
  the lemma.
\end{proof}

\subsubsection{{\sc Allocate}(2): Allocation algorithm for ``large'' sets}

The
allocation for the case that $d_v \geq k / 8$ or $d_v = 0$ for all $v
\in V$ is performed by
Algorithm~\ref{algo:WeightedMultiUnitAllocateLarge}. Here, we iterate starting with $t = 1$. In each
iteration, a subset $H_t$ of all users is selected by going though the
remaining users in decreasing order of $\pi$. If for a user $v$ we
have $\sum_{v \in H_t} d_v \cdot \bar{w}(u, v) < k / 32$, it is added
to $H_t$. However, in this case the allocation is immediately carried
out in a direct way: Each user that is added to $H_t$ is allocated an
arbitrary set of $d_v$ channels, e.g. the first ones. This iteration
is repeated with the remaining users that did not get allocated
anything until every user $v \in V$ has been allocated $d_v$ channels
in one iteration $t$. Finally, the algorithm picks the best of the
allocations computed in any single iteration.

\begin{algorithm}
Set $V_0 := V$ and $t := 0$\;
\While{$V_t \neq \emptyset$}{
  \For{$u \in V_t$ in decreasing order of $\pi$ values}{
    \If{$\sum_{v \in H_t} d_v \cdot \bar{w}(u, v) < k / 32$}{
      Add $u$ to $H_t$ and set $S_v^t = \{1,\ldots,d_v\}$ \;
    }
}
Let $V_{t+1}$ be the set of users who have not been allocated anything and set $t := t + 1$\;
}
Return the best one of the allocations $S^1, S^2, S^3, \ldots$\;
\SetAlgoRefName{{\sc Allocate}(2)}
\caption{Channel allocation for users that require at least $k / 8$
  channels.\label{algo:WeightedMultiUnitAllocateLarge}}
\end{algorithm}

\begin{proposition}
The algorithm computes at most $O(\log n + \log k)$ allocations and all of them are feasible.
\end{proposition}

\begin{proof}
  Using exactly the same arguments as in Claim~\ref{clm:log} above, we observe
  \[
  \sum_{v \in H_t} d_v > \sum_{u \in V_t \setminus H_t} d_u \enspace,
  \]
  which shows that at most $O(\log n + \log k)$ allocations are computed.

  The allocations are feasible since the sum of incoming weights on
  any channel is bounded by
  \[
  \sum_{u \in H_t} \bar{w}(u, v) \leq \frac{8}{k} \sum_{u \in H_t} d_u \cdot
  \bar{w}(u, v) < \frac{8}{k} \cdot \frac{k}{32} = \frac{1}{4} \enspace.
  \]
\end{proof}


\subsection{Truthfulness}

To turn the approximation algorithms from the previous section into
truthful mechanisms, we follow the idea by Lavi and
Swamy~\cite{Lavi05} using the randomized meta-rounding
technique~\cite{Carr02} to obtain a MIDR mechanism.
Our approach is similar to the one for
general secondary spectrum auctions~\cite{Hoefer11}. Linear
program~\eqref{lp:MultiUnit} are standard
packing LPs that allow to set up a separation LP to decompose an
optimal fractional solution scaled down by some approximation factor
larger than the integrality gap. Via this new LP we derive a
decomposition into integral solutions that represent feasible
solutions for the channel allocation problem. An optimal solution to
the decomposition LP is a probability distribution, by which we can
randomly pick a feasible solution in the resulting mechanism. As
usual, charging scaled VCG payments results in a mechanism that is
truthful in expectation.

The decomposition LP uses exponentially many variables (probabilities
for every possible feasible solution) but only polynomially many
constraints (decomposition of each non-zero variable in the fractional
optimum of LP~\eqref{lp:MultiUnit}). Thus,
the dual of the LP can be solved using the ellipsoid method with a
suitable separation oracle. The latter can be constructed from our
algorithms presented in the last section, as they verify the correct
approximation factor used for scaling in the decomposition LP. At this
point it is important to remark that the algorithms were defined to be
randomized. Therefore, the running time of the ellipsoid method would
only be polynomial in expectation.  However, all of our algorithms can
be derandomized using standard techniques.  The randomization in
Algorithms~\ref{algo:MultiUnit} and
\ref{algo:WeightedMultiUnitLPRounding} only depends on pairwise
independence.  Algorithm~\ref{algo:WeightedMultiUnitAllocateSmall} can
be made deterministic by using a combination of pairwise-independence
and conditional-expectation techniques. Under these conditions, the
desired decomposition can be found in polynomial time.

A main
drawback of this method is that the dual variables of the
decomposition must be interpreted as valuations of a new channel
allocation problem. Here assumptions like symmetry or submodularity
cannot be made, and algorithms for such special classes of valuations
might not be applicable. However, in our case the symmetry assumption
is encoded directly into LPs~\eqref{lp:MultiUnit} by setting up variables for each
\emph{number} and each not \emph{set} of channels. This property
carries over to the decomposition dual and our algorithms can be
applied.

\section{Matroid-Rank-Sum Valuations}
\label{sec:mrs}

In this section, we treat the class of so-called \emph{matroid rank
  sum (MRS) valuations}, in which $b_v$ for each bidder is a weighted
sum of matroid rank functions. This covers all frequently considered
submodular valuation functions such as, e.g., coverage functions,
matroid weighted-rank functions, and any convex combinations of these.

For ordinary combinatorial auctions, Dughmi et al.~\cite{DughmiSTOC11}
present an MIDR mechanism. The range is given by all solutions to a
linear relaxation of the item-allocation problem. Rounding is done via
a non-standard randomized rounding scheme called \emph{Poisson
  rounding} in~\cite{DughmiSTOC11}. Finding the optimal distribution
implies finding the fractional allocation that will achieve best
social welfare in expectation in the rounding stage. The Poisson
scheme is a \emph{convex rounding scheme}, for which finding the best
fractional allocation becomes a convex program with objective function
being expected social welfare.

Unfortunately, the Poisson rounding scheme is tailored to fit to
ordinary combinatorial auctions. The rounding is performed item-wise
-- when $x_{i,j}$ is the fractional allocation of item $j$ to bidder
$i$, then $j$ is fully given to $i$ with probability $1 - \e^{-x_{v,
    j}}$. With the remaining probability no bidder receives
$j$. Unlike items, the channels in our case can be given to multiple
users, and it takes significantly more effort to build a convex
rounding scheme. In the following we present our approach for this
case. We follow the conventions in~\cite{DughmiSTOC11}, in particular,
with respect to representation of MRS valuations using lottery-value
oracles. In particular, we will show the following theorem.

\begin{theorem}
  There is a truthful mechanism for MRS valuations that runs in
  expected polynomial time and returns a feasible allocation
  representing a $O(\rho)$-approximation for unweighted and a $O(\rho
  \cdot \log n)$-approximation for edge-weighted conflict graphs.
\end{theorem}

\subsection{Defining the Range}
We define the distributional range \calD\ in this section and discuss
why it is sufficiently large to get good approximations. Our starting
point are all fractional solutions $x$ fulfilling the following linear
constraints:
\begin{subequations}
  \begin{align}
    \sum_{\substack{u \in V \\ \pi(u) < \pi(v)}} \bar{w}(u, v) \cdot x_{u, j} \leq \rho  && \text{for all $v \in V$, $j \in [k]$} \label{eq:mrs-constraint-indindnumb} \\
    0 \leq x_{v, j} \leq 1 && \text{for all $v \in V$, $j \in [k]$} \label{eq:mrs-constraint-box}
  \end{align}
\end{subequations}

\begin{algorithm}
\For{$j \in [k]$}{
Draw $p_j$ uniformly for $[0, 1]$\;
  Decompose $(x_{v, j})_{v \in V}$ such that $x = \frac{1}{\alpha} \sum_{l} \lambda_l g_l$ and $\sum_l \lambda_l = 1$\;
  Let $l'$ be the minimal $l$ for which $\sum_{l < l'} \lambda_l < p_j$\;
  Allocate $g_{l'}$ tentitavely\;
  Remove each $v \in V$ from solution with a further probability of $p_{v,j} = \frac{1 - \e^{- x_{v, j} / (2\alpha)} }{ \frac{x_{v, j}}{\alpha}}$\;
}
\caption{Rounding scheme for a given solution $x$.}
\label{alg:roundingexact}
\end{algorithm}

For each channel we pick a feasible independent set separately in our
rounding scheme Algorithm~\ref{alg:roundingexact}. For each channel
$j$ the corresponding fractional solution $x_{\cdot,j}$ is decomposed
into polynomially many independent sets using parameter $\alpha$
discussed below. The algorithm selects one of these at random. The
decomposition can be computed in polynomial time using randomized
meta-rounding~\cite{Carr02,Lavi05} in combination with an appropriate
rounding scheme. Afterwards, each user $v$ is removed from the
solution by an independent random experiment rendering the total
probability for $v$ to receive channel $j$ to be exactly $1 -
\e^{-x_{v, j} / 2\alpha}$. Note that $p_{v,j}$ must be a valid
probability with $p_{v,j} \in [0,1]$. Here we observe that since
numerator and denominator are both positive, $p_{v,j}$ also
is. $p_{v,j} \le 1$ because $1 - \e^{x_{v, j} / (2\alpha)} \leq
\frac{x_{v, j}}{2 \alpha}$, for any $\alpha \ge 1$. Consequently, the
range \calD\ is given by all probability distributions resulting from
our rounding scheme applied to fractional solutions
of~\eqref{eq:mrs-constraint-indindnumb}
and~\eqref{eq:mrs-constraint-box}.

We have to specify the parameter $\alpha$, which ensures that
the decomposition of $x_{\cdot,j}$ exists. We interpret
$x_{\cdot,j}$ as solution to a linear program to maximize $\sum_{v
  \in V} a_v \cdot x_{v,j}$ subject to the
constraints~\eqref{eq:mrs-constraint-indindnumb}
and~\eqref{eq:mrs-constraint-box} for channel $j$. This is essentially
a linear relaxation for a single channel allocation problem with some
valuations $a_v$. We denote by $\alpha$ the integrality gap of this
program with respect to \emph{feasible} independent sets (Note that
the constraints~\eqref{eq:mrs-constraint-indindnumb} allow integer
solutions $x$ that represent infeasible independent sets). For this
program we can verify an integrality gap of $\alpha = O(\rho \cdot
\log n)$ for feasible independent sets using, e.g., the LP-rounding
algorithm for edge-weighted conflict graphs from~\cite{Hoefer11}. For
unweighted conflict graphs, the simpler LP-rounding algorithm
from~\cite{Hoefer11} yields $\alpha = O(\rho)$. Here, the simple
greedy algorithm of~\cite{Akcoglu00} (for details see
Section~\ref{sec:discuss} below) can even be shown to yield $\alpha =
\rho$.

For application of the randomized metarounding framework, we need an
algorithm verifying an integrality gap $\alpha$. This allows to
construct a decomposition LP and its dual, which can be solved in
polynomial time using the ellipsoid method, where the algorithm acts
as separation oracle (for details on this method
see~\cite{Carr02,Lavi05}). Note that $\alpha$ can merely be seen as a
parameter that serves to scale a fractional solution $x$ into a region
where a decomposition into (feasible) integral solutions exists --
independent of any objective function. The reason we interpret it as
integrality gap of an optimization problem is that the dual of the
decomposition LP allows an approximation algorithm verifying the gap
to be used to separate the dual and derive the required decomposition
in polynomial time. The reason we do not simply radically overestimate
$\alpha$ is that it does play a central role when we discuss the
approximation factor of our rounding scheme.

For a given distribution, the expected social welfare of the returned
allocation is exactly
\begin{equation}
  \sum_{v \in V} \sum_{T \subseteq [k]} b_v(T) \prod_{j \in T} (1 - \e^{-x_{v, j} / (2\alpha)}) \prod_{j \not\in T} \e^{-x_{v, j} / (2\alpha)} \label{eq:mrs-objective} \enspace.
\end{equation}
For the case of MRS functions, this function is concave, as we will
observe in more detail below. Therefore, the best distribution in the
range can be arbitrarily approximated by solving a convex program,
maximizing the concave objective (\ref{eq:mrs-objective}) subject to
linear constraints (\ref{eq:mrs-constraint-indindnumb}) and
(\ref{eq:mrs-constraint-box}).

As previously mentioned, the size of the range affects approximation
factor and tractability. Concerning the approximation factor, we can
show that the social welfare of the optimal allocation is at most an
$O(\alpha)$-factor above the expected social welfare of the best
distribution in the range.

\begin{lemma}
  The optimal distribution within the range is $O(\alpha)$-approximate
  in expectation when valuations are submodular. Hence, in
  expectation, the solution of our rouding scheme is a
  $O(\rho)$-approximation for unweighted and a $O(\rho \cdot \log
  n)$-approximation for edge-weighted conflict graphs.
\end{lemma}

\begin{proof}
  The optimal allocation $S^\ast$ corresponds to a feasible solution
  $x^\ast$ of the convex program. However, $x^\ast$ is not always
  rounded to $S^*$ but also to worse allocations. We bound the
  expected welfare of the received allocation in terms of that of
  $S^*$. This then yields the upper bound on the approximation ratio.
  The probability of each user $v$ of being allocated channel $j$ in
  rounding is exactly $1 - \e^{-x^\ast_{v, j} / (2\alpha)}$. We denote
  $b(S^*) = \sum_{v \in V} b_v(S^*(v))$ and use Proposition~C.4 in
  \cite{DughmiSTOC11}. This yields an expected social welfare of the
  rounded allocation of at least $(1 - \e^{-1 / (2\alpha)}) \cdot
  b(S^*) \geq \left(1 - \e^{-1} \right) \cdot (2\alpha)^{-1} \cdot
  b(S^*)$ due to concavity. Thus, the result of rounding the best
  distribution is at most a factor of $O(\alpha)$ worse.
\end{proof}

\subsection{Sampling the MIDR Distribution}

The expected social welfare when rounding a fractional solution $x$ is
given by \eqref{eq:mrs-objective}. Fortunately, this function is
concave in terms of $x$ meaning an optimal fractional solution can be
approximated arbitrarily well in polynomial time. However, to make the
mechanism truthful in expectation, we are, in principle, required to
solve the given convex program \emph{exactly}.

Since this is not possible, Algorithm~\ref{alg:simulation} devises a
way to simulate an exact solution in expected polynomial time. It
returns an allocation in which each bidder has exactly the same probability as in
Algorithm~\ref{alg:roundingexact} to get a channel. It requires us to compute
$\delta$-estimates -- a solution $x$ of the convex program such that
$x^\ast_{v, j} - \delta \leq x_{v, j} \leq x^\ast_{v, j} + \delta$ for
all $v$, $j$. To simplify the presentation, we assume that this can be
computed in time $\poly(n, k, \log(1/\delta))$. For details on this
issue, see Section~\ref{sec:delta-estimates}.

\begin{algorithm}
\For{$j \in [k]$}{
Draw $p_j$ uniformly from $[0, 1)$ and let $r$ be the minimal $t$ for which $p_j \geq 1 - 2^{-t+1}$ \;
Set $x^0 = 0$\;
\For{t = 1, \ldots, r}{
Compute $\delta^t$-estimate $x^t$, where $\delta^t = 1 / (n\cdot 2^{t+1})$\;
Let $y^t_v = \max\{ y^{t-1}_v, x^t_{v, j} - \delta^t \}$\;
}
Decompose $y^r - y^{r-1}$ such that $y^r = \frac{1}{2 \alpha} \sum_l \lambda^{r, l} g^{r, l}$ with $\sum_l \lambda^{r, l} = 2^{-r}$\;
Let $l'$ be the minimum $l$ such that $p_j > 1 - 2^{-r-1} + \sum_{l < l'} \lambda_l$ \;
Tentatively allocate $g^{r, l'}$ \;
Remove each $v \in V$ from solution with further probability $p_{v,j} = \frac{2 \alpha \left( \e^{-y^{t-1}_v / (2\alpha)} - \e^{-y^t_v / (2\alpha)} \right)} {y^t_v - y^{t-1}_v}$
}

\caption{Simulating Algorithm~\ref{alg:roundingexact} with estimates
  of the optimal convex-program solution.}
\label{alg:simulation}
\end{algorithm}

\begin{proposition}
  The desired decomposition $(g^{r,l},\lambda^{r,l})_{l}$ exists and
  can be computed in polynomial time.
\end{proposition}

\begin{proof}
  We distinguish between the two cases $r = 1$ and $r \geq 2$.

  In the case of $r=1$, $y^r$ fulfills
  equations~\eqref{eq:mrs-constraint-indindnumb}
  and~\eqref{eq:mrs-constraint-box}. Here we can apply the
  decomposition as described above. Using the algorithms
  from~\cite{Hoefer11} verifying integrality gaps of $\alpha =
  O(\rho)$ or $\alpha=O(\rho \cdot \log n)$, we can solve the
  decomposition LP of the meta-rounding framework and decompose $y^b =
  \frac{1}{\alpha} \tilde{\lambda}^{r, l} g^{r, l}$ with $\sum_l
  \tilde{\lambda}^{r, l} = 1$ where $g^{r, l}$ are integral solutions
  corresponding to independent sets. The running time is polynomial in
  $n$ and $k$. Setting $\lambda^{r, l} = \frac{1}{2}
  \tilde{\lambda}^{r, l}$ for all $l$ yields the desired composition.

  For the case $r \geq 2$, we use the fact that $x^{r-1}$ is already a
  $1/(n2^r)$-estimate. This yields that $0 \leq y^r_v -
  y^{r-1}_v \leq 1/(n2^{r-1})$. Therefore, it is possible to decompose
  $y^r - y^{r-1}$ to the trivial single-vertex independent
  sets. Formally, we consider an arbitrary ordering of the users $v_1,
  \ldots, v_n$, e.g. the one given by $\pi$. We set $g^{r, l}_{v_l} =
  1$ and $g^{r, l}_v = 0$ if $v_l \neq v$. The weights are set to
  $\lambda^{r, l} = \frac{1}{2 \alpha} (y^r_{v_l} -
  y^{r-1}_{v_l})$. This yields that $\sum_{l=1}^n \lambda^{r, l} \leq
  \sum_{l=1}^n \frac{1}{2\alpha} \cdot \frac{1}{n 2^{r-1}} \leq 2^{-r}$. The
  remaining weight is assigned to the all-zero fractional solution.
\end{proof}

\begin{proposition}
  For the probability of being removed we have $p_{v,j}\in [0,1]$.
\end{proposition}

\begin{proof}
  Since $y^{t-1}_v \leq y^t_v$ for all $v \in V$, the probability is
  at least $0$. Furthermore, we have
  \begin{align*}
    \frac{2 \alpha \left( \e^{-y^{t-1}_v / (2\alpha)} - \e^{-y^t_v / (2\alpha)} \right)} {y^t_v - y^{t-1}_v} & = \frac{2 \alpha \e^{-y^{t-1}_v / (2\alpha)} \left( 1 - \e^{- ( y^t_v - y^{t-1}_v ) / (2\alpha)} \right)} {y^t_v - y^{t-1}_v} \\
    & \leq \frac{2 \alpha \left( 1 - \e^{- ( y^t_v - y^{t-1}_v ) / (2\alpha)} \right)} {y^t_v - y^{t-1}_v} \\
    & \leq 1 \enspace.
  \end{align*}
\end{proof}

\begin{proposition}
  For each user $v \in V$ and each channel $j \in [k]$ the probability
  to receive $j$ is exactly $1 - \e^{-x^\ast_{v, j} / (2\alpha)}$.
\end{proposition}

\begin{proof}
  Let $r$ be defined as in the algorithm. Let us first consider the
  conditional probability of getting the channel given that $r = t$
  for some $t$.
  \begin{align*}
    \Pr{\text{$v$ gets $j$} \mid r = t} & = \Pr{g^{r, l'}_v = 1 \mid r = t} \cdot \frac{2 \alpha \left( \e^{-y^{t-1}_v / (2\alpha)} - \e^{-y^t_v / (2\alpha)} \right)} {y^t_v - y^{t-1}_v} \\
    & = \frac{2^t \left( y^{t-1}_v - y^t_v \right)}{2 \alpha} \cdot \frac{2 \alpha \left( \e^{-y^{t-1}_v / (2\alpha)} - \e^{-y^t_v / (2\alpha)} \right)} {y^t_v - y^{t-1}_v} \\
    & = 2^t \left( \e^{-y^{t-1}_v / (2\alpha)} - \e^{-y^t_v /
        (2\alpha)} \right)
  \end{align*}
  We get
  \[
  \Pr{\text{$v$ gets $j$}} = \sum_{t = 1}^\infty \Pr{r = t} \cdot
  \Pr{\text{$v$ gets $j$} \mid r = t} = \sum_{t = 1}^\infty 2^{-t}
  \cdot 2^t \left( \e^{-y^{t-1}_v / (2\alpha)} - \e^{-y^t_v /
      (2\alpha)} \right) = 1 - \e^{-x^\ast_{v, j}} \enspace,
  \]
  where the last step is due to the fact that $y^t_v$ converges to
  $x^\ast_{v, j}$ as $t \to \infty$.
\end{proof}

\begin{proposition}
  Assuming that the $\delta$-estimates can be computed in time
  $\poly(n, k, \log(1/\delta))$, the expected running time of
  Algorithm~\ref{alg:simulation} is polynomial in $n$ and $k$.
\end{proposition}

\begin{proof}
  Let us first consider the running time for the case that $r = t$ for
  some fixed $t$. If this case the $\delta$-estimates in lines 5--7
  can be computed in time $\sum_{i=1}^t \poly(n, k, \log(2^{i+1} n)) =
  \poly(n, k, t)$. The remaining computations take time $\poly(n,
  k)$. As a consequence, the expected running time of the algorithm is
  $\sum_{t=1}^\infty \Pr{r = t} \cdot \poly(n, k, t) =
  \sum_{t=1}^\infty 2^{-t} \cdot \poly(n, k, t) = \poly(n, k)$, where
  the last step is due to a geometric series.
\end{proof}

\subsection{Computing $\delta$-Estimates}
\label{sec:delta-estimates}

Algorithm~\ref{alg:simulation} only runs in expected polynomial time
when assuming that a $\delta$-estimate of the convex program can be
computed in time $\poly(n, k, \log(1/\delta))$. The reasoning why we
assume this is essentially the same as
in~\cite{DughmiSTOC11}. However, for the sake of completeness, we
present the most important steps in this section.

First of all, we have to observe that the objective function is
concave when all player valuations are MRS.

\begin{lemma}
  Our rounding scheme is convex when player valuations are MRS.
\end{lemma}

\begin{proof}
  Due to $\Ex{\sum_{v \in V} b_v(S_v)} = \sum_{v \in V}
  \Ex{b_v(S_v)}$, the result follows when $\Ex{b_v(S_v)}$ is concave
  for all $v$. By construction the probability for each user to be
  allocated channel $j$ is exactly $1 - \e^{-x_{v, j} /
    (2\alpha)}$. Therefore each $\Ex{b_v(S_v)}$ can be written as
  \[
  \sum_{T \subseteq [k]} b_v(T) \prod_{j \in T} (1 - \e^{-x_{v, j} /
    (2\alpha)}) \prod_{j \not\in T} \e^{-x_{v, j} / (2\alpha)}
  \enspace.
  \]
  We only have to prove that this function is concave over $(0,1)^k$.

  Dughmi et al.~\cite{DughmiSTOC11} show that the function $G\colon
  \RR^k \to \R$ with
  \[
  G(x_1, \ldots, x_k) = \sum_{T \subseteq [k]} b(T) \prod_{j \in T} (1
  - \e^{-x_j}) \prod_{j \not\in T} \e^{-x_j}
  \]
  is concave over $x \in (0,1)^k$ when $b$ is MRS.

  For $\Ex{b_v(S_v)} = G(x / (2\alpha))$ this also yields concavity
  since for any $\xi \in [0,1]$
  \[
  G\left( \frac{\xi x + (1 - \xi) y}{(2\alpha)} \right) =
  G\left( \xi\frac{x}{(2\alpha)} + (1 - \xi)
    \frac{y}{(2\alpha)} \right) \geq \xi G\left(
    \frac{x}{(2\alpha)} \right) + (1 - \xi)
  G\left(\frac{y}{(2\alpha)} \right)\enspace.
  \]
\end{proof}

This immediately yields the following claim when taking into
consideration that the constraints are linear.

\begin{claim}
  There is an algorithm in the lottery-value oracle model that, given
  an instance of spectrum auctions with edge-weighted conflict graphs
  on $n$ bidders and $k$ channels and an approximiation parameter
  $\epsilon > 0$, runs in $\poly(n, k, \log(1/\epsilon))$ time and
  returns a $(1-\epsilon)$-approximate solution to the convex
  program. 
\end{claim}

This yields the following result for $\delta$-estimates. Suppose we
are in the well-conditioned case, i.e., on any line in the feasible set
the second derivative of the objective function is at least $\lambda =
\frac{\sum_{v \in V} b_v([k])}{2^{\poly(n, k)}}$. Then a
$\delta$-estimate can be computed by computing an
$(1-\epsilon)$-approximate solution of the convex program with
$\epsilon=\frac{\delta^2}{2 \sum_{v \in V} b_v([k])}$. This solution
can be computed in time $\poly(n, k, \log(1/\delta))$.

\subsubsection{Guaranteeing Good Conditioning}
In general, the bound on the second derivative does not necessarily
have to hold. Therefore, the algorithm is modified as given in
Algorithm~\ref{alg:modifiedsimulation}.

\begin{algorithm}
Run Algorithm~\ref{alg:simulation}, let $S$ be the resulting allocation\;
Let $\beta$ be $\frac{1}{n k} \sum_{v \in V} \lvert S(v) \rvert$\;
Draw $q_1$ uniformly at random from $[0, 1]$\;
\If{$q_1 \leq \mu$}{
Set $S(v) := \emptyset$ for all $v \in V$\;
Draw $q_2$ uniformly at random from $[0, 1]$\;
\If{$q_2 \leq \beta$}{
Choose some user $v^\ast \in V$ uniformly at random\;
Set $S(v^\ast) = [k]$ and $S(v) = \emptyset$ for all $v \neq v^\ast$\;
}
}
\caption{Modified MIDR Algorithm.\label{alg:modifiedsimulation}}
\end{algorithm}

After having run Algorithm~\ref{alg:simulation}, the resulting
allocation is discarded with probability $\mu = 2^{-nk}$. Instead a
trivial allocation is returned, in which either only a single user
gets allocated all channels or even no channels are allocated at all,
as determined by another random experiment. However, since this action
is only taken with probability $1 - \mu = 1 - o(1)$, the approximation
factor is not affected.

On the contrary, we can show that the expected social welfare changes,
now having a curvature of at least $\lambda$. This is the missing
piece to build the $\delta$-estimates necessary to run the algorithm.

In order to determine the precise expected social welfare of the
modified algorithm, we have to first quantify the probability that the
initially computed solution is discarded. This is done with
probability $\beta$, which depends on the previous outcome. For the
expectation, we know
\[
\Ex{ \beta } = \Ex{ \frac{1}{nk} \sum_{v \in V} \lvert S(v) \rvert } =
\frac{1}{n k} \sum_{j = 1}^k \sum_{v \in V} \Pr{ j \in S(v) } =
\frac{1}{n k} \sum_{j = 1}^k \sum_{v \in V} \left( 1 - \e^{x_{v, j} /
    (2\alpha)} \right) \enspace.
\]
Therefore the expected social welfare is
\begin{eqnarray*}
&& (1 - \mu) \cdot \Ex{ b(S) } + \mu \cdot \Ex{ \beta } \frac{1}{n} \sum_{v \in V}
b_v([k])\\
& =& (1 - \mu) \cdot \Ex{ b(S) } + \frac{\mu}{n^2 k} \cdot \left( \sum_{j =
    1}^k \sum_{v \in V} \left( 1 - \e^{x_{v, j} / (2\alpha)} \right)
\right) \cdot \left( \sum_{v \in V} b_v([k]) \right) \enspace.
\end{eqnarray*}
Since both parts of the outer sum are non-negative, it suffices to
bound the curvature of the second one. The curvature of $\left(
  \sum_{j = 1}^k \sum_{v \in V} \left( 1 - \e^{x_{v, j} / (2\alpha)}
  \right) \right)$ is at least $(\e (2\alpha)^2)^{-1}$. Therefore, the
curvature of the second part is at least
\[
\frac{\mu}{n^2 k} \cdot \frac{1}{\e (2\alpha)^2} \cdot \left( \sum_{v \in V}
  b_v([k]) \right) = \lambda \enspace.
\]
As a consequence, the modified algorithm can be run with
$\delta$-approximates as described above with a resulting running time
that is $\poly(n,k)$ in expectation.


\section{Discussion and Open Problems}
\label{sec:discuss}

While the mechanisms presented in previous sections obtain
near-optimal guarantees on social welfare, they have some drawbacks
for application in practice. A serious problem are running times --
for MRS valuations our mechanism obtains polynomial running time only
in expectation. For symmetric valuations, we obtain polynomial
worst-case running times, but the convex optimization techniques
needed to apply randomized meta-rounding often have prohibitive
running times for large practical problem instances. Thus, let us
briefly discuss designing fast and simple mechanisms. How can we design a good and simple deterministic
mechanism to incentivize truth-telling among bidders?

To our knowledge, there are only two algorithmic approaches to the
channel assignment problem that yield approximation guarantees in the
order of $O(\rho)$. One approach is rounding of suitably relaxed
packing LPs, which turned out to be very successful in this and our
previous work~\cite{Hoefer11}. While pairwise independence can be used
to make these algorithms deterministic, they require randomization to
guarantee truthfulness and fail for deterministic truthfulness. The
other approach was proposed for the simplest case of a single channel
and unweighted conflict graphs, i.e., the maximum weighted independent
set problem. It is a simple greedy algorithm due to Akcoglu et
al~\cite{Akcoglu00} which first considers vertices one by one in
reverse of the ordering of $\pi$. If vertex $v$ is under
consideration, its current value is subtracted from the value of each
backward neighbor. If the value of a vertex drops to 0 or below before
it is under consideration in the ordering, this vertex is
removed. Finally, the algorithm makes a second pass over the surviving
vertices, this time in forward ordering of $\pi$, and greedily adds
each vertex to the independent set if possible. It can be shown using
a local ratio argument that it provides a
$\rho$-approximation~\cite{Ye09}.

It is tempting to believe that this algorithm is monotone and delivers
a deterministically truthful mechanism. Unfortunately, this is not the
case, see our example in Figure~\ref{fig:NotMonotone}. The problem is
that the algorithm makes a second pass over the vertices which
introduces non-trivial dependencies among bids and acceptance
decisions. Nevertheless, we show how to turn it into a monotone
algorithm by spending a $\log n$ factor in the approximation
guarantee. This is a promising first step towards designing simple
truthful deterministic mechanisms with non-trivial approximation
guarantees. In contrast to algorithms using the time-intensive
solution of convex optimization problems, such quick and simple greedy
rules are much more suitable for application in practice. Providing
good and simple mechanisms is a major open direction for future work.

\begin{figure}
  \begin{center}
\begin{tikzpicture}[scale=1.3]
  \tikzstyle{vertex}=
  [%
    draw=black,%
    minimum size=2mm,%
    circle,%
    thick%
  ]
\node at (0, -2) {(a)};
\node[vertex,label=above:$x$] (1) at (0,0) {$7$};
\node[vertex,label=above:$6$] (2) at (1,0) {$6$};
\node[vertex,label=above:$6$] (3) at (2,0) {$5$};
\node[vertex,label=right:$7$] (4) at (1,-1) {$4$};
\node[vertex,label=right:$7$] (5) at (2,-1) {$3$};
\node[vertex,label=right:$7$] (6) at (3,-1) {$2$};
\node[vertex,label=right:$11.5$] (7) at (1,-2) {$1$};
\path (1) edge (2);
\path (1) edge (7);
\path (2) edge (3);
\path (2) edge (4);
\path (3) edge (4);
\path (3) edge (5);
\path (3) edge (6);
\path (4) edge (7);
\path (5) edge (7);
\path (6) edge (7);
\end{tikzpicture}
\begin{tikzpicture}[scale=1.3]
  \tikzstyle{vertex}=
  [%
    draw=black,%
    minimum size=2mm,%
    circle,%
    thick%
  ]
  \tikzstyle{vertextaken}=
  [%
    draw=black,%
    fill=black!40,%
    minimum size=2mm,%
    circle,%
    thick%
  ]
\node at (0, -2) {(b)};
\node[vertextaken,label=above:$3$] (1) at (0,0) {$7$};
\node[vertex,label=above:$3$] (2) at (1,0) {$6$};
\node[vertex,label=above:$3$] (3) at (2,0) {$5$};
\node[vertextaken,label=right:$1$] (4) at (1,-1) {$4$};
\node[vertextaken,label=right:$4$] (5) at (2,-1) {$3$};
\node[vertextaken,label=right:$4$] (6) at (3,-1) {$2$};
\node[vertex,label=right:$-0.5$] (7) at (1,-2) {$1$};
\path (1) edge (2);
\path (1) edge (7);
\path (2) edge (3);
\path (2) edge (4);
\path (3) edge (4);
\path (3) edge (5);
\path (3) edge (6);
\path (4) edge (7);
\path (5) edge (7);
\path (6) edge (7);
\end{tikzpicture}
\begin{tikzpicture}[scale=1.3]
  \tikzstyle{vertex}=
  [%
    draw=black,%
    minimum size=2mm,%
    circle,%
    thick%
  ]
  \tikzstyle{vertextaken}=
  [%
    draw=black,%
    fill=black!40,%
    minimum size=2mm,%
    circle,%
    thick%
  ]
\node at (0, -2) {(c)};
\node[vertex,label=above:$4$] (1) at (0,0) {$7$};
\node[vertex,label=above:$2$] (2) at (1,0) {$6$};
\node[vertextaken,label=above:$4$] (3) at (2,0) {$5$};
\node[vertex,label=right:$1$] (4) at (1,-1) {$4$};
\node[vertex,label=right:$3$] (5) at (2,-1) {$3$};
\node[vertex,label=right:$3$] (6) at (3,-1) {$2$};
\node[vertextaken,label=right:$0.5$] (7) at (1,-2) {$1$};
\path (1) edge (2);
\path (1) edge (7);
\path (2) edge (3);
\path (2) edge (4);
\path (3) edge (4);
\path (3) edge (5);
\path (3) edge (6);
\path (4) edge (7);
\path (5) edge (7);
\path (6) edge (7);
\end{tikzpicture}
  \end{center}
  \caption{\label{fig:NotMonotone} Example for non-monotonicity of the
    greedy algorithm. In part (a), the number inside the circle denotes the vertex's index in the $\pi$-ordering, the one outside its reported valuation. If bidder$7$ reports $x=3$, he is included in the solution; if he bids up to $x = 4$, he is dropped by the
    algorithm. Parts (b) and (c) depict the resulting values and the
    independent sets at the end of the execution of the algorithm for
    each case, respectively.}
\end{figure}
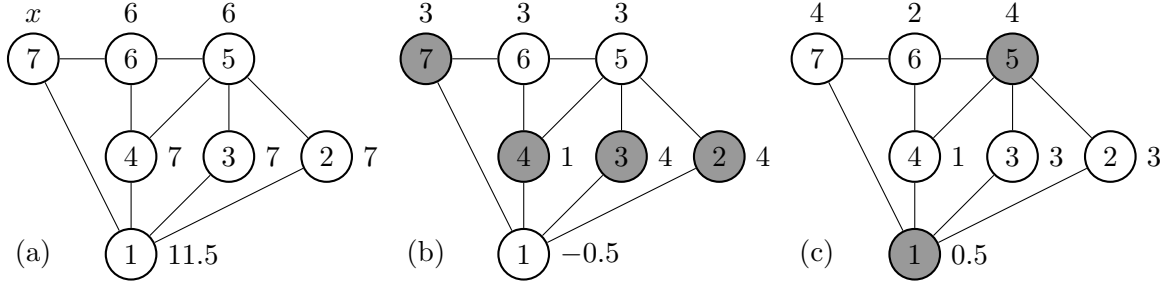

\begin{algorithm}
Sort the set of bids $B = \{ b_v \mid v \in V\}$ in decreasing order,
let $b_i$ be the $i$-th highest bid\; %
\For{$i=1$ to $n$}{ %
  Let $V_i = \{ v \in V \mid b_v \ge b_i \}$ and $S_i = \emptyset$\; %
  \For{$v \in V_i$ in increasing order of $\pi$ values}{ %
    If $N(v) \cap S_i = \emptyset$, add $v$ to $S_i$ } }%
Output $S = S_{i^*}$ with $i^* = \arg\max_{i} \sum_{v \in S_i}
b_v$\;
\caption{Monotone $O(\rho \cdot \log n)$-algorithm for Maximum
  Weighted Independent Set.\label{algo:greedy}}
\end{algorithm}

\begin{theorem}
  Algorithm~\ref{algo:greedy} is deterministic and monotone. The
  computed solution is a $O(\rho \cdot \log n)$-approximation for the
  maximum weight independent set problem.
\end{theorem}

\begin{proof}
  We first prove that the algorithm is monotone. We show that if $v
  \not\in S$ and lies a value $b'_v < b_v$, then $v$ will never be
  able to become part of $S$. Suppose $v$ is currently first
  considered in iteration $i$. Submitting a smaller bid causes $v$ to
  be considered at a later point $j > i$. In the sets
  $V_i,\ldots,V_{j-1}$ player $v$ is replaced by a different player,
  sets $V_1,\ldots,V_{i-1}$ and $V_j,\ldots,V_n$ remain as before, and
  so do $S_1,\ldots,S_{i-1}$ and $S_j,\ldots,S_n$. If one of these
  sets was chosen as the best set before, then $v$ will again not be
  part of $S$ if he lies. The only sets that can be different now are
  $S_i,\ldots,S_{j-1}$, in which $v$ cannot be present. If previously
  set $S_k$ with $k \in \{i,\ldots,j-1\}$ was chosen as the best set,
  it did not include $v$. Thus, in the run $v$ was blocked by some
  other vertex. Removing $v$ does not change the execution of the
  algorithm, thus the same set will be computed again -- however, due
  to the change in the ordering it will now appear as $S_{k-1}$. The
  only sets $S_k$ that can change are the ones with $k \in
  \{i,\ldots,j-1\}$ where $v$ was included before. However, if a new
  optimal set appears here, it does not include $v$ as well. In
  conclusion, if $v \not\in S$, he cannot become included into $S$ by
  reducing his bid.

  To bound the approximation factor, we use an argument similar
  to~\cite{Halldorsson00}. Let us consider the problem on the subset
  $V_i$ and assume all vertices have value $b_i$. For this problem,
  our algorithm is equivalent to the greedy $\rho$-approximation
  algorithm for unweighted vertices. Hence, for this subproblem we
  obtain a $\rho$-approximation. With $S_i'$ being the optimum for
  this subproblem, then we have
  \[\sum_{v\in S} b_v \quad = \quad \max_{i=1}^n \sum_{v \in S_i} b_v
  \quad \ge \quad \max_{i=1}^n
  \{|S_i|\cdot b_i\} \quad \ge \quad \frac{1}{\rho} \cdot \max_{i=1}^n
  \{|S_i'|\cdot b_i\}\enspace.\]
  Now consider intervals $I_j = (b_1/2^j,b_1/2^{j-1}]$, for
  $j=1,\ldots,\log n$. The last interval we set $I_{(\log n)+1} =
  [0,b_1/n]$. For each such interval we consider the subgraph of
  vertices $v$ with value $b_v \in I_j$ and the optimum solution $S^j$
  w.r.t. to this subinstance. Consider all $i$ such that $b_i \in
  I_j$. It is easy to see that for all $j=1,\ldots,\log n$
  \[ \frac{1}{\rho} \cdot \max_{i : b_i \in I_j} \{|S_i'|\cdot b_i\} \ge
  \frac{1}{2 \cdot \rho} \sum_{v \in S^j} b_v\enspace.\]
  For $j = (\log n)+1$ we obviously have $|S_1'|\cdot b_1 \ge \sum_{v
    \in S^j} b_v$. Thus, in total we have
  \[
  \frac{1}{\rho} \cdot \max_{i=1}^n \{|S_i'|\cdot b_i\} \quad \ge \quad
  \frac{1}{2\cdot\rho\cdot \log n + \rho} \cdot \sum_{j=1}^{\log n} \sum_{v
    \in S^j} b_v \quad \ge \quad \frac{1}{2\cdot\rho\cdot \log n + \rho} \cdot \sum_{v
    \in S^*} b_v\enspace,
  \]
  since the sum of values for the optimal solutions in the intervals
  is bigger than the global optimum $S^*$. This proves the
  approximation factor.
\end{proof}

This represents a promising first step towards designing simple
truthful deterministic mechanisms with non-trivial approximation
guarantees. In contrast to algorithms using the time-intensive
solution of convex optimization problems, such quick and simple greedy
rules are much more suitable for application in practice. In addition,
the concept of truthfulness in expectation used in the previous
sections has drawbacks, e.g., it is not enough to motivate
\emph{risk-aware} bidders to reveal their valuations truthfully. While
there are many open problems stemming from our work (e.g., improving
the approximation bounds for specific interference models), providing
good and simple mechanisms for stronger notions of truthfulness is a
challenging and arguably the most interesting avenue for future work.


\bibliographystyle{plain}
\bibliography{../../../Bibfiles/literature,../../../Bibfiles/martin}

\begin{thebibliography}{10}

\bibitem{Akcoglu00}
Karhan Akcoglu, James Aspnes, Bhaskar DasGupta, and Ming-Yang Kao.
\newblock Opportunity cost algorithms for combinatorial auctions.
\newblock {\em CoRR}, cs.CE/0010031, 2000.

\bibitem{Berry10}
Randall Berry, Michael Honig, and Rakesh Vohra.
\newblock Spectrum markets: {M}otivation, challenges, and implications.
\newblock {\em IEEE Communications Magazine}, 2010.

\bibitem{BlumrosenChapter07}
Liad Blumrosen and Noam Nisan.
\newblock Combinatorial auctions.
\newblock In Noam Nisan, {\'E}va Tardos, Tim Roughgarden, and Vijay Vazirani,
  editors, {\em Algorithmic Game Theory}, chapter~11. Cambridge University
  Press, 2007.

\bibitem{Briest05}
Patrick Briest, Piotr Krysta, and Berthold V{\"o}cking.
\newblock Approximation techniques for utilitarian mechanism design.
\newblock In {\em Proc.\ 37th Symp.\ Theory of Computing (STOC)}, pages 39--48,
  2005.

\bibitem{Carr02}
Robert Carr and Santosh Vempala.
\newblock Randomized metarounding.
\newblock {\em Random Struct.\ Algorithms}, 20(3):343--352, 2002.

\bibitem{Chen10}
Danny Chen, Rudolf Fleischer, and Jian Li.
\newblock Densest-subgraph approximation on intersection graphs.
\newblock In {\em Proc.\ 8th Intl.\ Workshop Approximation and Online
  Algorithms (WAOA)}, pages 83--93, 2010.

\bibitem{Dobzinski07}
Shahar Dobzinski.
\newblock Two randomized mechanisms for combinatorial auctions.
\newblock In {\em Proc.\ 10th Intl.\ Workshop Approximation Algorithms for
  Combinatorial Optimization Problems (APPROX)}, pages 89--103, 2007.

\bibitem{Dobzinski09}
Shahar Dobzinski and Shaddin Dughmi.
\newblock On the power of randomization in algorithmic mechanism design.
\newblock In {\em Proc.\ 50th Symp.\ Foundations of Computer Science (FOCS)},
  pages 505--514, 2009.

\bibitem{DobzinskiCoRR10}
Shahar Dobzinski, Hu~Fu, and Robert Kleinberg.
\newblock Truthfulness via proxies.
\newblock {\em CoRR abs/1011.3232}, 2010.

\bibitem{Dobzinski10}
Shahar Dobzinski and Noam Nisan.
\newblock Mechanisms for multi-unit auctions.
\newblock {\em J. Artif.\ Intell.\ Res.}, 37:85--98, 2010.

\bibitem{Dughmi11}
Shaddin Dughmi.
\newblock A truthful randomized mechanism for combinatorial public projects via
  convex optimization.
\newblock In {\em Proc.\ 12th Conf.\ Electronic Commerce (EC)}, pages 263--272,
  2011.

\bibitem{DughmiSTOC11}
Shaddin Dughmi, Tim Roughgarden, and Qiqi Yan.
\newblock From convex optimization to randomized mechanims: {T}oward optimal
  combinatorial auctions.
\newblock In {\em Proc.\ 43rd Symp.\ Theory of Computing (STOC)}, pages
  149--158, 2011.

\bibitem{DughmiFOCS11}
Shaddin Dughmi and Jan Vondr{\'a}k.
\newblock Limitations of randomized mechanisms for combinatorial auctions.
\newblock In {\em Proc.\ 52nd Symp.\ Foundations of Computer Science (FOCS)},
  2011.
\newblock To appear.

\bibitem{Fanghaenel10}
Alexander Fangh\"anel, Sascha Geulen, Martin Hoefer, and Berthold V\"ocking.
\newblock Online capacity maximization in wireless networks.
\newblock In {\em Proc.\ 22nd Symp.\ Parallelism in Algorithms and
  Architectures (SPAA)}, pages 92--99, 2010.

\bibitem{GopinathanREV11}
Ajay Gopinathan and Zongpeng Li.
\newblock A prior-free revenue maximizing auction for secondary spectrum
  access.
\newblock In {\em Proc.\ 30th IEEE Conf.\ Computer Communications (INFOCOM)},
  pages 86--90, 2011.

\bibitem{Gopinathan11}
Ajay Gopinathan, Zongpeng Li, and Chuan Wu.
\newblock Strategyproof auctions for balancing social welfare and fairness in
  secondary spectrum markets.
\newblock In {\em Proc.\ 30th IEEE Conf.\ Computer Communications (INFOCOM)},
  pages 3020--3028, 2011.

\bibitem{Goussevskaia09}
Olga Goussevskaia, Magn{\'u}s Halld{\'o}rsson, Roger Wattenhofer, and Emo
  Welzl.
\newblock Capacity of arbitrary wireless networks.
\newblock In {\em Proc.\ 28th IEEE Conf.\ Computer Communications (INFOCOM)},
  pages 1872--1880, 2009.

\bibitem{Halldorsson00}
Magn{\'u}s Halld{\'o}rsson.
\newblock Approximations of weighted independent set and hereditary subset
  problems.
\newblock {\em J. Graph Alg.\ Appl.}, 4(1):1--16, 2000.

\bibitem{Halldorsson11}
Magn{\'u}s Halld{\'o}rsson and Pradipta Mitra.
\newblock Wireless capacity with oblivious power in general metrics.
\newblock In {\em Proc.\ 22nd Symp.\ Discrete Algorithms (SODA)}, pages
  1538--1548, 2011.

\bibitem{Hoefer11}
Martin Hoefer, Thomas Kesselheim, and Berthold V\"ocking.
\newblock Approximation algorithms for secondary spectrum auctions.
\newblock In {\em Proc.\ 23rd Symp.\ Parallelism in Algorithms and
  Architectures (SPAA)}, pages 177--186, 2011.

\bibitem{Kesselheim11}
Thomas Kesselheim.
\newblock A constant-factor approximation for wireless capacity maximization
  with power control in the {SINR} model.
\newblock In {\em Proc.\ 22nd Symp.\ Discrete Algorithms (SODA)}, pages
  1549--1559, 2011.

\bibitem{Kesselheim10}
Thomas Kesselheim and Berthold V{\"o}cking.
\newblock Distributed contention resolution in wireless networks.
\newblock In {\em Proc.\ 24th Intl.\ Symp.\ Distributed Computing (DISC)},
  pages 163--178, 2010.

\bibitem{Khot08}
Subhash Khot, Richard Lipton, Evangelos Markakis, and Aranyak Mehta.
\newblock Inapproximability results for combinatorial auctions with submodular
  utility functions.
\newblock {\em Algorithmica}, 52(1):3--18, 2008.

\bibitem{Lavi05}
Ron Lavi and Chaitanya Swamy.
\newblock Truthful and near-optimal mechanism design via linear programming.
\newblock In {\em Proc.\ 46th Symp.\ Foundations of Computer Science (FOCS)},
  pages 595--604, 2005.

\bibitem{Lehmann02}
Daniel Lehmann, Liadan O'Callaghan, and Yoav Shoham.
\newblock Truth revelation in approximately efficient combinatorial auctions.
\newblock {\em J. ACM}, 49(5), 2002.

\bibitem{Mirrokni08}
Vahab Mirrokni, Michael Schapira, and Jan Vondr{\'a}k.
\newblock Tight information-theoretic lower bounds for welfare maximization in
  combinatorial auctions.
\newblock In {\em Proc.\ 9th Conf.\ Electronic Commerce (EC)}, pages 70--77,
  2008.

\bibitem{Mualem02}
Ahuva Mu'alem and Noam Nisan.
\newblock Truthful approximation mechanisms for restricted combinatorial
  auctions.
\newblock In {\em Proc.\ 18th Conf.\ Artificial Intelligence (AAAI)}, pages
  379--384, 2002.

\bibitem{Papadimitriou08}
Christos Papadimitriou, Michael Schapira, and Yaron Singer.
\newblock On the hardness of being truthful.
\newblock In {\em Proc.\ 49th Symp.\ Foundations of Computer Science (FOCS)},
  pages 250--259, 2008.

\bibitem{Trevisan01}
Luca Trevisan.
\newblock Non-approximability results for optimization problems on bounded
  degree instances.
\newblock In {\em Proc.\ 33rd Symp.\ Theory of Computing (STOC)}, pages
  453--461, 2001.

\bibitem{Voecking12}
Berthold V\"ocking.
\newblock A universally-truthful approximation scheme for multi-unit auctions.
\newblock In {\em Proc.\ 23rd Symp.\ Discrete Algorithms (SODA)}, 2012.
\newblock To appear.

\bibitem{Vondrak08}
Jan Vondr{\'a}k.
\newblock Optimal approximation for the submodular welfare problem in the value
  oracle model.
\newblock In {\em Proc.\ 40th Symp.\ Theory of Computing (STOC)}, pages 67--74,
  2008.

\bibitem{Wan09}
Peng-Jun Wan.
\newblock Multiflows in multihop wireless networks.
\newblock In {\em Proc.\ 10th Symp.\ Mobile Ad Hoc Networking and Computing
  (MobiHoc)}, pages 85--94, 2009.

\bibitem{WanWASA09}
Peng-Jun Wan, Xiaohua Jia, and F.~Frances Yao.
\newblock Maximum independent set of links under physical interference model.
\newblock In {\em Proc.\ 4th Intl.\ Conf.\ Wireless Algorithms, Systems,
  Applications (WASA)}, pages 169--178, 2009.

\bibitem{Ye09}
Yuli Ye and Allan Borodin.
\newblock Elimination graphs.
\newblock In {\em Proc.\ 36th Intl.\ Coll.\ Automata, Languages and Programming
  (ICALP)}, pages 774--785, 2009.

\bibitem{Zhou08}
Xia Zhou, Sorabh Gandhi, Subhash Suri, and Haitao Zheng.
\newblock {eBay in the Sky: S}trategy-proof wireless spectrum auctions.
\newblock In {\em Proc.\ 14th Intl.\ Conf.\ Mobile Computing and Networking
  (MOBICOM)}, pages 2--13, 2008.

\bibitem{Zhou09}
Xia Zhou and Haitao Zheng.
\newblock {TRUST}: {A} general framework for truthful double spectrum auctions.
\newblock In {\em Proc.\ 28th IEEE Conf.\ Computer Communications (INFOCOM)},
  pages 999--1007, 2009.

\end{thebibliography}

\end{document}